\documentclass[10pt]{article}
\usepackage{amsfonts}
\usepackage{amssymb,amsmath,amsthm,latexsym}
\usepackage{amsmath,amsthm,amssymb}
\newtheorem{Theorem}{Theorem}[section]
\newtheorem{lem}[Theorem]{Lemma}
\newtheorem{Remark}[Theorem]{Remark}

\newtheorem{Corollary}[Theorem]{Corollary}

\newtheorem{Example}[Theorem]{Example}
\numberwithin{equation}{section}

\begin{document}

\title{Construction of  Self-dual Codes over $F_p+vF_p$
\footnote{E-mail addresses: zgh09@yahoo.com.cn (G. Zhang),
b\_c\_chen@yahoo.com.cn (B. Chen)}}

\author{Guanghui Zhang$^1$, Bocong Chen$^2$}

\date{\small 1. School of Mathematical Sciences,
Luoyang Normal University,\\
Luoyang, Henan, 471022, China \\
2. School of Mathematics and Statistics,
Central China Normal University,\\
Wuhan, Hubei, 430079, China
}

\maketitle

\begin{abstract}
In this paper, we
determine all  self-dual codes over $F_p+vF_p$ ($v^2=v$)
in terms of self-dual codes over the finite field $F_p$
and
give an explicit construction for self-dual codes over $F_p+vF_p$,
where $p$ is  a prime.
\medskip


\end{abstract}
\section{Introduction}
Codes over finite rings were initiated in the early 1970s \cite{Blake1972, Blake1975}.
They have received much attention  since the seminal work \cite{HK},
which showed that certain good nonlinear
binary codes could be found as images of linear codes over $\mathbb{Z}_4$ under the Gray map.

Generally, most of the studies are concentrated on the situation when
the ground rings associated with codes are finite chain rings
( e.g. see \cite{DINH},\cite{Dinh10},\cite{GH},\cite{KL}-\cite{LB},\cite{W1} ).
However, it has been proved that finite Frobenius rings are suitable for coding alphabets \cite{WJ},
which leads to many works on codes over non-chain rings.

In recent years, linear codes over the ring $F_p+vF_p$ with $v^2=v$ and $p$ being a prime,
which is not a chain ring but a  Frobenius  ring,  have been considered by some authors.
In \cite{ZWS}, Zhu et al. gave some results on cyclic codes over $F_2+vF_2$,
where it is shown that cyclic codes over the ring are principally generated.
In \cite{ZW},  Zhu et al. studied $(1-2v)$-constacyclic codes over $F_p+vF_p$, where $p$ is an odd prime.
They determined the image of  $(1-2v)$-constacyclic codes over $F_p+vF_p$ under the Gray map
and the structures of such constacyclic codes over $F_p+vF_p$.
In \cite{CDD}, Cengellenmis et al. generated the ring $F_2+vF_2$ to the infinite family of rings
$A_k=F_2[v_1, v_2, \cdots,v_k]/\langle v_i^2=v_i, v_iv_j=v_jv_i\rangle, 1\leq i,j \leq k$,
and studied codes over these rings by using  Gray maps.

On the other hand, self-dual codes play a very significant role in coding theory
both from  practical and  theoretical points of
view. A vast number of papers have been devoted to the study of self-dual codes;
e.g. see \cite{DHS}-\cite{DKKL}, \cite{GL}, \cite{HP}-\cite{KL1}, \cite{MS,PH}. 
In \cite{KL}, Kim and Lee gave an efficient method to construct self-dual codes over finite
fields from a given self-dual code of a smaller length.
In \cite{DKKL}, Dougherty et al.
proved that self-dual codes exist over all finite commutative Frobenius rings and gave some building-up
constructions for self-dual codes over these rings.
More recently, Cengellenmis et al. \cite{CDD} studied Euclidean and Hermitian self-dual codes over
$A_k=F_2[v_1, v_2, \cdots,v_k]/\langle v_i^2=v_i, v_iv_j=v_jv_i\rangle, 1\leq i,j \leq k$,
and gave a sufficient and necessary condition for the existence of self-dual codes over the rings.

In this paper, we
determine all  self-dual codes over $F_p+vF_p$ ($v^2=v$)
in terms of self-dual codes over the finite field $F_p$
and give an explicit construction for self-dual codes over $F_p+vF_p$,
where $p$ is  a prime.
Unlike the technique used in the mentioned papers, we first give the parity check matrices for
linear codes over $F_p+vF_p$. Then we characterize the torsion codes associated with
the linear codes, which
are used as a tool to study self-dual codes over $F_p+vF_p$ and their explicit construction.

The organization of this paper is as follows.
The necessary notations and  some known results
are provided in Section 2.
In Section 3,  we first characterize the torsion codes, and then give some criteria for a linear code over the ring to be self-dual.
In Section 4, we determine all  self-dual codes over $F_p+vF_p$
and give an explicit construction for self-dual codes over $F_p+vF_p$.
In section 5, we give some examples to illustrate our main results.

\section{Preliminaries}
Let $F_p$ be a finite field with $p$ elements, where $p$ is a prime.
Throughout this paper,  $R$ denotes the commutative ring $F_p+vF_p=\{a+vb\,|\,a,b\in F_p\}$ with $v^2=v$.
Any element of $R$ can be uniquely expressed as $c=a+vb$, where $a,b\in F_p$.
The Gray map $\Phi$ from $R$ to $F_p\times F_p$ is given by $\Phi(c)=(a,a+b)$. It is routine to check that $\Phi$ is a ring isomorphism,
which means  $R$ is isomorphic to the ring $F_p\times F_p;$ so $R$ is a finite Frobenius ring.
The ring $R$ is a semi-local ring with exactly two maximal ideals given by $\langle v\rangle=\{av|a\in F_p\}$
and $\langle 1-v\rangle=\{a(1-v)|a\in F_p\}$.
It is easy to verify that both $R/\langle v\rangle$ and $R/\langle 1-v\rangle$ are isomorphic to $F_p$.

A code $C$ of length $n$ over $R$ is a nonempty subset of $R^n$, and the ring $R$ is
referred to the alphabet of the code. If this subset is also an $R$-submodule of $R^n$,
then $C$ is called linear. For any  code $C$ of length $n$ over $R$, the {\it dual code of $C$}
is defined as $C^\perp=\{u\in R^n\,|\,u\cdot v=0, \mbox{for any $v\in C$}\}$,
where $u\cdot v$ denotes the standard Euclidean inner product of $u$ and $v$ in $R^n$.
Notice that $C^\perp$ is linear whether or not $C$ is linear.
If $C\subseteq C^\perp$, then $C$ is called {\it self-orthogonal}. If $C = C^\perp$,
then $C$ is called {\it self-dual}.

We have known that the ring $R$ has exactly two maximal ideals $\langle v\rangle$ and $\langle 1-v\rangle$.
Their residue fields are both $F_p$. Thus we have two canonical projections defined as follows:
$$R=F_p+vF_p \longrightarrow R/\langle v\rangle = F_p$$
$$a+vb\longmapsto a;$$
and
$$R=F_p+vF_p \longrightarrow R/\langle 1-v\rangle = F_p$$
$$a+vb\longmapsto a+b.$$
We simple denote these two projections by ``\,\,$-$\,\,'' and``\,\,\,$\widehat{}$\,\,\,'', respectively.
Denote by $\overline{r}$ and $\widehat{r}$ the images of an element $r\in R$ under these two projections, respectively.

Note that any element $c$ of $R^n$ can be uniquely expressed as $c=r+vq$, where $r,q \in F_p^n$.
Let $C$ be a linear code of length $n$ over $R$. Define
$$C_1=\{a\in F_p^n|a+vb\in C, \text{for some} \, \,b\in F_p^n\}$$
and
$$
C_2=\{a+b\in F_p^n|a+vb\in C\}.
$$
Obviously, $C_1$ and $C_2$ are linear codes over $F_p$.

Assume that $a\in R$.
For a code $C$ of length $n$ over $R$,
the {\it submodule quotient} is a linear code of length $n$ over $R$, defined as follows:
$$(C:a)=\{x\in R^n|ax\in C\}.$$
The codes $\widehat{(C:v)}$ and $\overline{(C:(1-v))}$ over the field $F_p$ is called
the {\it torsion codes} associated with the code $C$ over the ring $R$.

For the case of odd prime $p$, any nonzero linear code $C$ over $R$ is permutation-equivalent to a
code generated by the following matrix ( see \cite{ZW}):
$$G=
\begin{pmatrix}
I_{k_1} & (1-v)B_1 & vA_1 & vA_2+(1-v)B_2 & vA_3+(1-v)B_3 \\
0 & vI_{k_2} & 0 & vA_4 & 0 \\
0 & 0 & (1-v)I_{k_3} & 0 & (1-v)B_4
\end{pmatrix},
$$
where $A_i$ and $B_j$ are matrices with entries in $F_p$ for $i,j=1,2,3,4$.
Such a code $C$ is said to have type $p^{2k_1}p^{k_2}p^{k_3}$ and $|C|=p^{2k_1+k_2+k_3}$.
For later convenience the above generator matrix can be rewritten in the form:
$$G=
\begin{pmatrix}
I_{k_1} & (1-v)B_1 & vA_1 & vD_1+(1-v)D_2  \\
0 & vI_{k_2} & 0 & vC_1 \\
0 & 0 & (1-v)I_{k_3} & (1-v)C_2
\end{pmatrix},\qquad \qquad \qquad(\ast)
$$
where $D_1=(A_2 \mid A_3), D_2=(B_2 \mid B_3), C_1=(A_4 \mid 0), C_2=(0 \mid B_4)$.

For the case  $p=2$, a nonzero linear code $C$ over $R$ has
a generator matrix which after a suitable permutation of the coordinates
can be written in the form ( see \cite{WZX,ZWS} ):
$$
G=
\begin{pmatrix}
I_{k_1} & A & B & D_1+vD_2\\
0 & vI_{k_2} & 0 & vC_1\\
0 & 0 & (1+v)I_{k_3} & (1+v)E
\end{pmatrix},\qquad \qquad \qquad(\ast)
$$
where $A, B, C_1, D_1, D_2$ and $E$ are  matrices  with entries in $F_2$,
and $|C|=2^{2k_1}2^{k_2}2^{k_3}$.

For $k>0$, $I_k$ denotes the $k\times k$ identity matrix.
The code $C_1$ is permutation-equivalent to a code with generator matrix of the form ( see \cite{ZW,ZWS} ):
\begin{equation*}
G_1=
\begin{cases}
\begin{pmatrix}
I_{k_1} & B_1 & 0 & B_2 & B_3 \\
0 & 0 & I_{k_3} & 0 & B_4
\end{pmatrix}, & \text{$p$ is odd;}\\
\begin{pmatrix}
I_{k_1} & A & B & D_1\\
0 & 0 & I_{k_3} & E
\end{pmatrix}, & \text{$p=2$,}
\end{cases}
\end{equation*}
where $A, B, E, D_1$ and $B_i$ are $p$-ary matrices for $i\in \{1,2,3,4\}$.
And the code $C_2$ is permutation-equivalent to a code with generator matrix of the form ( see \cite{ZW,ZWS} ):
\begin{equation*}
G_2=
\begin{cases}
\begin{pmatrix}
I_{k_1} & 0 & A_1 & A_2 & A_3 \\
0 & I_{k_2} & 0 & A_4 & 0
\end{pmatrix}, & \text{$p$ is odd;}\\
\begin{pmatrix}
I_{k_1} & A & B & D_1+D_2\\
0 & I_{k_2} & 0 & C_1
\end{pmatrix}, & \text{$p=2$,}
\end{cases}
\end{equation*}
where $A,B, C_1, D_1, D_2$ and $A_i$ are $p$-ary matrices for $i\in \{1,2,3,4\}$.
It is easy to see that ${\rm dim} C_1=k_1+k_3$ and ${\rm dim} C_2=k_1+k_2$.
\section{Self-dual codes over $F_p+vF_p$}
We begin with a lemma about the torsion codes associated with the code over the ring $R$,
which will be used throughout the paper.
\begin{lem}\label{mainlemma}
Assume the notation given above. Let $C$ be a linear code of length $n$ over $R$. Then

{\rm (1)} $\widehat{(C:v)}=C_2$.

{\rm (2)} $\overline{(C:(1-v))}=C_1$.
\end{lem}
\begin{proof}
(1) For any $y\in \widehat{(C:v)}$, there exists an $x\in (C:v)$ such that $y=\widehat{x}$.
Let $x=r+vq$, where $r,q\in F_p^n$. Then $\widehat{x}=r+q$. Since $vx\in C$, we have
$$v(r+q)=v(r+vq)=vx\in C,$$
which implies that $r+q\in C_2$. Therefore $y=\widehat{x}=r+q\in C_2$.
It follows that $\widehat{(C:v)}\subseteq C_2$.

Let $z\in C_2$. Then there exists an element $x+vy\in C$ such that $z=x+y$.
Hence
$$v(x+y)=v(x+vy)\in C,$$
and $x+y\in (C:v)$. Thus we have that
$$z=x+y=\widehat{x+y}\in \widehat{(C:v)}.$$
Hence $\widehat{(C:v)}\supseteq C_2$. Therefore we get the desired result.

(2) Let $y$ be an element of $\overline{(C:(1-v))}$, then there exists some $x \in (C:(1-v))$
such that $y=\overline{x}$. Suppose that $x=r+vq$, for $r, q\in F_p^n$. Then $\overline{x}=r$.
From $(1-v)x \in C$ we have that
$$r-vr=(1-v)r=(1-v)(r+vq)=(1-v)x\in C,$$
which leads to $r\in C_1$. Hence $y=\overline{x}=r\in C_1$. Therefore we obtain that $\overline{(C:v)}\subseteq C_1$.

If $r$ is an element of $C_1$, then we have that $r+vq \in C$ for some $q\in F_q^n$. Since
$$(1-v)r=(1-v)(r+vq)\in C,$$
which shows that $r\in (C:(1-v))$. Hence $r=\overline{r}\in \overline{(C:(1-v))}$,
then $\overline{(C:(1-v))}\supseteq C_1$. Therefore $\overline{(C:(1-v))}= C_1$, as required.
\end{proof}

In the following $A^T$ denotes the transpose of the matrix $A$.
Suppose $\mathcal{C}_1$ and $\mathcal{C}_2$ are permutation equivalent linear codes over $R$ with
$\mathcal{C}_1P=\mathcal{C}_2$ for some permutation matrix $P$. Then
$\mathcal{C}_1^\perp P=\mathcal{C}_2^\perp$. Without loss of generality,
we may assume that
a linear code $C$ of length $n$ over $R$ is with generator matrix in the form $(\ast)$.
\begin{Theorem}\label{generatormatrix}
Let $C$ be a linear code of length $n$ over $R$ with generator matrix in the form $(\ast)$.

{\rm (1)} For $p$ being odd, let
$$
H=
\begin{pmatrix}
v E_1+(1-v)E_2 & P & Q & I_{n-k}\\
v(-A_1^T) & 0 & v I_{k_3} & 0\\
(1-v)(-B_1^T) & (1-v)I_{k_2} & 0 & 0
\end{pmatrix},
$$
where $E_1=(-A_2 \mid A_1B_4-A_3)^T$, $E_2=(B_1A_4-B_2 \mid -B_3)^T$, $P=(-A_4 \mid 0)^T$,
$Q=(0 \mid -B_4)^T$ and   $k=k_1+k_2+k_3$. Then $H$ is a generator matrix for $C^\perp$ and a parity check matrix for $C$.

{\rm (2)} For $p=2$, let
$$
H=
\begin{pmatrix}
E^TB^T+C_1^TA^T+(D_1+vD_2)^T & C_1^T & E^T & I_{n-k}\\
vB^T & 0 & vI_{k_3} & 0 \\
(1+v)A^T & (1+v)I_{k_2} & 0 & 0
\end{pmatrix},
$$
where $A, B, C_1, D_1, D_2$ and $E$ are  matrices  with entries in $F_2$  and $k=k_1+k_2+k_3$.
Then $H$ is a generator matrix for $C^\perp$ and a parity check matrix for $C$.

{\rm (3)} $(\widehat{(C:v)})^\perp=\widehat{(C^\perp:v)}; (\overline{(C:(1-v))})^\perp=\overline{(C^\perp:(1-v))}$.
\end{Theorem}
\begin{proof}
(1) Since the verification of  $HG^T=0$ is routine and somewhat tedious, we present a detail proof in the
appendix.
Let $D$ be the $R$-submodule generated by $H$, then $D\subseteq C^\perp$.
Since $R$ is a Frobenius ring, we have  $|C||C^\perp|=|R|^n$ (\cite{WJ}). It follows that
$$|C^\perp|=\frac{|R|^n}{|C|}=\frac{p^{2n}}{p^{2k_1+k_2+k_3}}=p^{2(n-k_1)-k_2-k_3}.$$
Note that $|D|=p^{2(n-k)+k_3+k_2}=p^{2(n-k_1)-k_2-k_3}$, and we obtain  $|D|=|C^\perp|$,
hence $D=C^\perp$.

(2) Similar to the proof of (1).

(3) We first prove that $\widehat{(C^\perp:v)}\subseteq (\widehat{(C:v)})^\perp$.
Let $x\in (C^\perp:v)$ and $y\in (C:v)$. Then $vx\in C^\perp$ and $vy\in C$,
so $(vx)(vy)^T=0$, i.e., $v(xy^T)=0$. Hence $xy^T\in (1-v)R$, and $\widehat{x}\widehat{y}^T=0$,
which implies that $\widehat{(C^\perp:v)}\subseteq (\widehat{(C:v)})^\perp$.
On the other hand, by Lemma \ref{mainlemma} and Theorem \ref{generatormatrix}\,(1)(2), we have that
$${\rm dim}\widehat{(C^\perp:v)}=n-k+k_3=n-k_1-k_2;$$
$${\rm dim}\widehat{(C:v)}^\perp=n-{\rm dim}\widehat{(C:v)}=n-(k_1+k_2)=n-k_1-k_2.$$
Hence ${\rm dim}\widehat{(C^\perp:v)}={\rm dim}\widehat{(C:v)}^\perp$,
which follows that $(\widehat{(C:v)})^\perp=\widehat{(C^\perp:v)}$.

The proof of the second equality is similar to that of the first one and is omitted here.
\end{proof}
\begin{Corollary}\label{conditions-1}
Let $C$ be a linear code of length $n$ over $R$. Then
$C$ is self-dual if and only if both the following two conditions are satisfied:

{\rm (i)} $C$ is self-orthogonal;

{\rm (ii)} $n=2(k_1+k_2), k_2=k_3$.
\end{Corollary}
\begin{proof}
Now suppose that both Conditions (i) and (ii) are satisfied. Then we have that
$$|C|=p^{2k_1+k_2+k_3}=p^{2(k_1+k_2)}, |C^\perp|=p^{2(n-k)+k_2+k_3}=p^{2(k_1+k_2)}.$$
Note that $C\subseteq C^\perp$, and then $C = C^\perp$, that is, $C$ is self-dual.

Suppose that $C$ is self-dual, then $C$ is self-orthogonal.
By Lemma \ref{mainlemma} and Theorem \ref{generatormatrix}(1)(2), we have that
$${\rm dim}\widehat{(C:v)}=k_1+k_2;$$
$${\rm dim}\widehat{(C^\perp:v)}=n-k+k_3=n-k_1-k_2,$$
and
$${\rm dim}\overline{(C:(1-v))}=k_1+k_3;$$
$${\rm dim}\overline{(C^\perp:(1-v))}=n-k+k_2=n-k_1-k_3.$$
Since $C=C^\perp$, we have that $n=2(k_1+k_2), k_2=k_3$.
\end{proof}

Let $A, B$ be the codes over $R$. We denote by $A\oplus B=\{a+b\,|\,a\in A, b\in B\}$.
\begin{Theorem}\label{maindecomposition}
With the above notations, let $C$ be a linear code of length $n$ over $R$.
Then $C$ can be uniquely expressed as $C=vC_2\oplus (1-v)C_1$.
Moreover, we also have $C^\perp=vC^\perp_2\oplus (1-v)C^\perp_1$.
\end{Theorem}
\begin{proof}
We first prove the uniqueness of the expression of every element in $vC_2\oplus (1-v)C_1$.
Let $va_2+(1-v)a_1=vb_2+(1-v)b_1$, where $a_2, b_2 \in C_2$ and $a_1, b_1 \in C_1$.
Then $v(a_2-b_2)=(1-v)(b_1-a_1)$, which implies that $a_1=b_1$ and $a_2=b_2$.
Hence $|vC_2\oplus (1-v)C_1|=|C_1||C_2|=p^{k_1+k_3}p^{k_1+k_2}=p^{2k_1+k_2+k_3}=|C|$.

Next we prove that $vC_2\oplus (1-v)C_1\subseteq C$. Let $a\in (C:v)$ and $b\in (C:(1-v))$.
Then $va \in C$ and $(1-v)b\in C$. Assume  $a=a_1+(1-v)a_2, b=b_1+vb_2$,
where $a_1, a_2, b_1, b_2 \in F_p^n$. Then $\widehat{a}=a_1\in C_2, \overline{b}=b_1\in C_1$. Thus
$$v\widehat{a}+(1-v)\overline{b}=va_1+(1-v)b_1=va+(1-v)b\in C.$$
Hence $vC_2\oplus (1-v)C_1\subseteq C$. Note that $|vC_2\oplus (1-v)C_1|=|C|$,
therefore $C=vC_2\oplus (1-v)C_1$.

Finally, we prove  the second statement. Combining the first statement, Theorem \ref{generatormatrix}(3)
with Lemma \ref{mainlemma} we have
\begin{eqnarray*}
C^\perp & = & v\widehat{(C^\perp:v)}\oplus (1-v)\overline{(C^\perp:(1-v))}\\
        & = & v(\widehat{(C:v)})^\perp\oplus (1-v)(\overline{(C:(1-v))})^\perp\\
        & = & vC^\perp_2\oplus (1-v)C^\perp_1,
\end{eqnarray*}
which is the desired result.
\end{proof}
\begin{Corollary}\label{maincor}
With the above notations, let $C$ be a linear code of length $n$ over $R$.
Then $C$ is a self-dual code if and only if $C_1$ and $C_2$ are both self-dual codes.
\end{Corollary}
\begin{proof}
$(\Longrightarrow)$ Let $C$ be a self-dual code. Then by Lemma \ref{mainlemma} and Theorem \ref{generatormatrix}(3)
we have
$$C^\perp_1=(\overline{(C:(1-v)})^\perp=\overline{(C^\perp:(1-v))}=\overline{(C:(1-v))}=C_1$$
and
$$C^\perp_2=(\widehat{(C:v)})^\perp=\widehat{(C^\perp:v)}=\widehat{(C:v)}=C_2,$$
that is, $C_1$ and $C_2$ are both self-dual codes.

$(\Longleftarrow)$ Let $C_1$ and $C_2$ be both self-dual codes. Then by Theorem \ref {maindecomposition},
$$C^\perp=vC^\perp_2\oplus (1-v)C^\perp_1=vC_2\oplus (1-v)C_1=C.$$
So $C$ is self-dual.
\end{proof}

\begin{Remark}\label{expression-1}
According to Theorem \ref{maindecomposition} and Corollary \ref{maincor},
it is clear that a self-dual code over $R$ can be explicitly expressed via
two self-dual codes over $F_p$. We need to study the converse part,
which is an interesting step.
\end{Remark}

\section{Construction of self-dual codes over $F_p+vF_p$}
The construction of self-dual codes over $R$ depends on the following theorem.
\begin{Theorem}\label{selfdual}
Suppose that $\mathcal{C}_1$ and $\mathcal{C}_2$ are linear codes of length $n$ over $F_p$
with generator matrices $G_1$ and $G_2$ respectively, and let $l_1$ and $l_2$ be
the dimensions of $\mathcal{C}_1$ and $\mathcal{C}_2$ respectively.
Then the code $C$ over $R$ generated by the matrix $G$,

\begin{equation*}
G =
\begin{cases}
\begin{pmatrix}
vG_2\\0
\end{pmatrix}
+(1-v)G_1, & \text{if $l_1> l_2$};\\
vG_2+
\begin{pmatrix}
(1-v)G_1\\0
\end{pmatrix}, & \text{if $l_1< l_2$};\\
v G_2+(1-v)G_1, & \text{if $l_1 =l_2$}.
\end{cases}
\end{equation*}
satisfies
$$C=v\mathcal{C}_2\oplus (1-v)\mathcal{C}_1,~~ \widehat{(C:v)}=\mathcal{C}_2,  \,\,\, \overline{(C:(1-v))}=\mathcal{C}_1.$$
\end{Theorem}
\begin{proof}
We only prove the case  $l_1>l_2$ in the following, as the proof of the other cases are similar to this case.
Assume that $G_1=(g_{11}, g_{12},\cdots,g_{1,l_1})^T$, $G_2=(g_{21}, g_{22},\cdots,g_{2,l_2})^T$,
then
$$G=
\begin{pmatrix}
vg_{21}+(1-v)g_{11} \\
 vg_{22}+(1-v)g_{12}\\
 \vdots\\
 vg_{2,l_2}+(1-v)g_{1,l_2}\\
  (1-v)g_{1,l_2+1}\\
  \vdots\\
  (1-v)g_{1,l_1}
\end{pmatrix}.
$$

Since $vg_{2i}+(1-v)g_{1i} \in C$, i.e. $g_{1i}+v(g_{2i}-g_{1i})\in C$, for every $1\leq i\leq l_2$,
by Lemma \ref{mainlemma} we have
$$g_{2i}=g_{1i}+(g_{2i}-g_{1i})\in \widehat{(C:v)},$$
for every $1\leq i\leq l_2$. Therefore $\mathcal{C}_2\subseteq \widehat{(C:v)}$.

Let $y\in \widehat{(C:v)}$, then there exists $x\in (C:v)$ such that $y=\widehat{x}$.
Since $vx \in C$, we may assume that
$$vx=\sum_{i=1}^{l_2}(a_i+vs_i)[vg_{2i}+(1-v)g_{1i}]+\sum_{l_2+1}^{l_1}(a_i+vs_i)[(1-v)g_{1i}],$$
where $a_i+vs_i\in F_p+vF_p$, for $1\leq i \leq l_1$. So
$$vx=v^2x=v\cdot vx=v\sum_{i=1}^{l_2}(a_i+s_i)g_{2i}.$$

Let $x=x_1+vx_2, x_1, x_2\in F_p^n$. Then $\widehat{x}=x_1+x_2$. Thus
$$v(x_1+x_2)=vx=v\sum_{i=1}^{l_2}(a_i+s_i)g_{2i}.$$
Hence $x_1+x_2=\sum_{i=1}^{l_2}(a_i+s_i)g_{2i}$. Therefore we have
$$y=\widehat{x}=x_1+x_2=\sum_{i=1}^{l_2}(a_i+s_i)g_{2i} \in \mathcal{C}_2,$$
which gives $\widehat{(C:v)}\subseteq \mathcal{C}_2$. From the above facts we get that $\widehat{(C:v)}= \mathcal{C}_2$.

On the other hand, since
$$vg_{2i}+(1-v)g_{1i} \in C, i.e. \;\; g_{1i}+v(g_{2i}-g_{1i})\in C,$$
for every $1\leq i\leq l_1$.  Here $g_{2i}=0$, if $i> l_2$.
By Lemma \ref{mainlemma} we have $g_{1i}\in \overline{(C:(1-v))}$,
for every $1\leq i\leq l_1$. Therefore $\mathcal{C}_1\subseteq \overline{(C:(1-v))}$.

Let $z\in \overline{(C:(1-v))}$, then there exists $s\in (C:(1-v))$ such that $z=\overline{s}$.
Since $(1-v)s \in C$, we  assume that
$$(1-v)s=\sum_{i=1}^{l_2}(b_i+vt_i)[vg_{2i}+(1-v)g_{1i}]+\sum_{l_2+1}^{l_1}(b_i+vt_i)[(1-v)g_{1i}],$$
where $b_i+vt_i\in F_p+vF_p$, for $1\leq i \leq l_1$. So
$$(1-v)s=(1-v)^2s=(1-v)\cdot (1-v)s=(1-v)\sum_{i=1}^{l_1}b_ig_{1i}.$$

Let $s=s_1+vs_2, s_1, s_2\in F_p^n$. Then $\overline{s}=s_1$. Thus
$$(1-v)s_1=(1-v)s=(1-v)\sum_{i=1}^{l_1}b_ig_{1i}.$$
Hence $s_1=\sum_{i=1}^{l_1}b_ig_{1i}$. Therefore we have
$$z=\overline{s}=s_1=\sum_{i=1}^{l_1}b_ig_{1i} \in \mathcal{C}_1,$$
which gives $\overline{(C:(1-v))}\subseteq \mathcal{C}_1$. Thus we get  $\overline{(C:(1-v))}= \mathcal{C}_1$.

Finally, by Lemma \ref{mainlemma} and Theorem \ref{maindecomposition},
\begin{eqnarray*}
C & = & v \widehat{(C:v)} \oplus (1-v)\overline{(C:(1-v))}\\
 & = & v\mathcal{C}_2 \oplus (1-v)\mathcal{C}_1,
\end{eqnarray*}
which gives our desired result. Thus we complete the proof.
\end{proof}
\begin{Corollary}\label{construction}
Suppose that $\mathcal{C}_1$ and $\mathcal{C}_2$ are two self-dual codes of length $n$ over $F_p$
with generator matrices $G_1$ and $G_2$ respectively,
then the code $C$ over $R$ generated by the matrix $G$ as follows
is also self-dual, where
$$
G = v G_2+(1-v)G_1.
$$
\end{Corollary}
\begin{proof}
Note that $l_1=l_2$ in this case. By Lemma \ref{mainlemma},  Theorem \ref{maindecomposition}
and Theorem \ref{selfdual} we have
\begin{eqnarray*}
C^\perp & = & v(\widehat{(C:v)})^\perp\oplus (1-v)(\overline{(C:(1-v))})^\perp\\
        & = & v\mathcal{C}^\perp_2\oplus (1-v)\mathcal{C}^\perp_1\\
        & = & v\mathcal{C}_2\oplus (1-v)\mathcal{C}_1\\
        & = & C.
\end{eqnarray*}
So $C$ is self-dual.
\end{proof}

\begin{Theorem}\label{allcodes}
All the self-dual codes over $R$ are given by
$$v \mathcal{C}_2\oplus (1-v)\mathcal{C}_1,$$
where $\mathcal{C}_1, \mathcal{C}_2$ range over all the self-dual codes over $F_p$, respectively.
Moreover, this expression is unique, i.e. if
$$v \mathcal{C}_2\oplus (1-v)\mathcal{C}_1 = v \mathcal{C}'_2\oplus (1-v)\mathcal{C}'_1,$$
then $\mathcal{C}_2 = \mathcal{C}'_2$ and $\mathcal{C}_1 = \mathcal{C}'_1$,
where $\mathcal{C}_1, \mathcal{C}_2, \mathcal{C}'_1 $ and $\mathcal{C}'_2$ are all
self-dual codes over $F_p$.
\end{Theorem}
\begin{proof}
First by Corollary~\ref{maincor}, every self-dual code over $R$
can be explicitly expressed by two fixed self-dual codes over $F_p$ as in the above form.

Next, let $\mathcal{C}_1, \mathcal{C}_2$ be arbitrary two self-dual codes over $F_p$.
Assume that $G_1$ and $G_2$ are generator matrices for $\mathcal{C}_1, \mathcal{C}_2$, respectively.
Then according to Corollary \ref{construction} we know that the code $C$ generated by the matrix
$vG_1 + (1-v)G_2$ is self-dual and satisfies $C = v \mathcal{C}_2\oplus (1-v)\mathcal{C}_1$.
This completes the proof of the first statement.

Let $x\in \mathcal{C}_2$. Since $v \mathcal{C}_2\oplus (1-v)\mathcal{C}_1 = v \mathcal{C}'_2\oplus (1-v)\mathcal{C}'_1$,
we have that
$$vx\in v \mathcal{C}_2 \subseteq v \mathcal{C}_2\oplus (1-v)\mathcal{C}_1 = v \mathcal{C}'_2\oplus (1-v)\mathcal{C}'_1.$$
Assuming  $vx =vx'+(1-v)y'$ where $x'\in \mathcal{C}'_2, y'\in \mathcal{C}'_1$, we get that
$v(x-x')=(1-v)y'$ and $v(x-x')=0$, so $x=x'$. Therefore $\mathcal{C}_2\subseteq \mathcal{C}'_2$.
Similarly, we have $\mathcal{C}'_2\subseteq \mathcal{C}_2$. Hence $\mathcal{C}_2 = \mathcal{C}'_2$.

Let $z\in \mathcal{C}_1$. Since $v \mathcal{C}_2\oplus (1-v)\mathcal{C}_1 = v \mathcal{C}'_2\oplus (1-v)\mathcal{C}'_1$,
we have that
$$(1-v)z\in (1-v) \mathcal{C}_1 \subseteq v \mathcal{C}_2\oplus (1-v)\mathcal{C}_1 = v \mathcal{C}'_2\oplus (1-v)\mathcal{C}'_1.$$
Setting $(1-v)z =va'+(1-v)z'$, where $a'\in \mathcal{C}'_2, z'\in \mathcal{C}'_1$, we get that
$(1-v)(z-z')=va'$ and $(1-v)(z-z')=0$, so $z=z'$. Therefore $\mathcal{C}_1\subseteq \mathcal{C}'_1$.
Similarly, we have $\mathcal{C}'_1\subseteq \mathcal{C}_1$. Hence $\mathcal{C}_1 = \mathcal{C}'_1$.
Thus we complete the proof.
\end{proof}
\begin{Corollary}
Let $N(R)$ be the number of self-dual codes of length $n$ over $R$
and $N(F_p)$ the number of self-dual codes of length $n$ over $F_p$.
Then
$$N(R) = N(F_p)^2.$$
\end{Corollary}
\begin{proof}
It follows immediately from Theorem \ref{allcodes}.
\end{proof}

The following lemma is well known and can be found from \cite{RS}.
\begin{lem}\label{conditions-2}
Let $F_q$ be a finite field with characteristic $p$. Then

{\rm (i)} If $p=2$ or $p\equiv 1\,({\rm mod}\,4)$, then a self-dual code of length $n$ exists over $F_q$
if and only if $n\equiv 0\,({\rm mod}\,2)$;

{\rm (ii)} If $p\equiv 3\,({\rm mod}\,4)$, then a self-dual code of length $n$ exists over $F_q$
if and only if $n\equiv 0\,({\rm mod}\,4)$.
\end{lem}

Now Combining  Theorem \ref{allcodes} with Corollary \ref{conditions-2}, the following result is easily obtained.
\begin{Theorem}
With the above notations. Then the following two statements hold:

{\rm (i)} If $p=2$ or $p\equiv 1\,({\rm mod}\,4)$, then a self-dual code of length $n$ over $R$ exists
if and only if $n\equiv 0\,({\rm mod}\,2)$;

{\rm (ii)} If $p\equiv 3\,({\rm mod}\,4)$, then a self-dual code of length $n$ over $R$ exists
if and only if $n\equiv 0\,({\rm mod}\,4)$.
\end{Theorem}
\begin{Remark}
For $p=2$, the corresponding result has been obtained in \cite[Corollary 5.5]{CDD}.
\end{Remark}

\section{Examples}
According to Corollary \ref{construction}, the construction of self-dual codes over $R$
hinges on constructing the self-dual codes over $F_p$. See \cite{KL} on the building-up
construction of self-dual codes over $F_p$.
The following examples illustrate our results.
\begin{Example}
Consider the construction of self-dual code of length $4$ over $R=F_5+vF_5$.
Let $c=2$ be in $F_5$ such that $c^2=-1$ in $F_5$. Here $l_1=l_2=2$ and
$$
G_1=
\begin{pmatrix}
1 & 0 & 3 & 0\\
-3 & 1 & 1 & 2
\end{pmatrix};
$$
$$
G_2=
\begin{pmatrix}
0 & 2 & 0 & 1\\
-2 & 4 & 1 & 2
\end{pmatrix}.
$$
Then the code $C$ of length $4$ over $R=F_5+vF_5$ generated by the following matrix
$$
G=vG_2+(1-v)G_1=
\begin{pmatrix}
1-v & 2v & 3-3v & v \\
-3+v & 1+3v & 1 & 2
\end{pmatrix}
$$
is self-dual.

On the other hand, it is an elementary calculation to check that the above code $C$
is permutation-equivalence to a code $\mathcal{C}$ generated by the following matrix:
$$
\begin{pmatrix}
1 & 0 & 2+v & 0\\
0 & 1 & 0 & 2+v
\end{pmatrix}
=
\begin{pmatrix}
I_2 \mid vD_1+(1-v)D_2
\end{pmatrix},
$$
where $D_1=
\begin{pmatrix}
3 & 0\\
0 & 3
\end{pmatrix},
D_2=
\begin{pmatrix}
2 & 0\\
0 & 2
\end{pmatrix}
$.
By the Corollary \ref{conditions-1}, it is easy to check that $\mathcal{C}$ is self-dual.
So the code $C$ is also self-dual.
\end{Example}
\begin{Example}
Consider the construction of self-dual code of length $6$ over $R=F_2+vF_2$.
Here $l_1=l_2=3$ and
$$
G_1=
\begin{pmatrix}
1 & 0 & 1 & 1 & 0 & 1\\
1 & 1 & 1 & 0 & 1 & 0\\
1 & 1 & 1 & 1 & 1 & 1
\end{pmatrix};
$$
$$
G_2=
\begin{pmatrix}
1 & 0 & 0 & 1 & 0 & 0\\
0 & 0 & 1 & 0 & 0 & 1\\
1 & 1 & 1 & 1 & 1 & 1
\end{pmatrix}.
$$
Then the code $C$ of length $6$ over $R=F_2+vF_2$ generated by the following matrix
\begin{eqnarray*}
G & = & v G_2+(1-v)G_1 \\
& = & G_1+v(G_2-G_1) \\
& = &
\begin{pmatrix}
1 & 0 & 1+v & 1 & 0 & 1+v\\
1+v & 1+v & 1 & 0 & 1+v & v\\
1 & 1 & 1 & 1 & 1 & 1
\end{pmatrix}
\end{eqnarray*}
is self-dual.

Similarly, throught an elementary calculation, the above code $C$
is permutation-equivalence to a code $\mathcal{C}$ generated by the following matrix:
$$
\begin{pmatrix}
1 & 0 & 0 &  v & 0 & 1+v\\
0 & 1 & 0 & 1+v & 0 & v \\
0 & 0 & 1 & 0 & 1 & 0
\end{pmatrix}
=
\begin{pmatrix}
I_3 \mid D_1+vD_2
\end{pmatrix},
$$
where $D_1=
\begin{pmatrix}
0 & 0 & 1\\
1 & 0 & 0\\
0 & 1 & 0
\end{pmatrix},
D_2=
\begin{pmatrix}
1 & 0 & 1\\
1 & 0 & 1\\
0 & 0 & 0
\end{pmatrix}
$.
By the Corollary \ref{conditions-1}, it is easy to check that $\mathcal{C}$ is self-dual.
Thus the code $C$ is also self-dual.
\end{Example}
\begin{Example}
Consider the construction of self-dual code of length $12$ over $R=F_3+vF_3$.
Here $l_1=l_2=6$ and
$$G_1=(I_6\mid B),$$
where $I_6$ denotes the $6\times 6$ identity matrix, and
$$
B=
\begin{pmatrix}
0 & 1 & 1 & 1 & 1 & 1\\
1 & 0 & 1 & 2 & 2 & 1\\
1 & 1 & 0 & 1 & 2 & 2\\
1 & 2 & 1 & 0 & 1 & 2\\
1 & 2 & 2 & 1 & 0 & 1\\
1 & 1 & 2 & 2 & 1 & 0
\end{pmatrix},
$$
i.e. the code with generator matrix $G_1$ is the ternary Golay code;
$$
G_2=\left(\begin{array}{llllccccrrrrr}
0 & 1 & 1 & 1 & 0 & 0 & 0 & 0 & 0 & 0 & 0 & 0 \\
1 & 0 & 0 & 0 & 1 & 0 & 1 & 2 & 0 & 1 & 1 & 0 \\
0 & 0 & 0 & 0 & 0 & 1 & 1 & 1 & 0 & 0 & 0 & 0 \\
0 & 0 & 0 & 0 & 1 & 0 & 0 & 0 & 1 & 0 & 2 & 0 \\
0 & 0 & 0 & 0 & 0 & 0 & 0 & 0 & 1 & 2 & 1 & 0 \\
2 & 1 & 2 & 0 & 1 & 2 & 1 & 0 & 2 & 2 & 0 & 1
\end{array}\right).
$$
Then the code $C$ of length $12$ over $R=F_3+vF_3$ generated by the following matrix
\begin{eqnarray*}
G & = & v G_2+(1-v)G_1 \\
& = & G_1+v(G_2-G_1) =  \\
\end{eqnarray*}
$$
\left(\begin{array}{llllccccrrrrr}
1+2v & v & v & v & 0 & 0 & 0 & 1+2v & 1+2v & 1+2v & 1+2v & 1+2v \\
v & 1+2v & 0 & 0 & v & 0 & 1 & 2v & 1+2v & 2+2v & 2+2v & 1+2v \\
0 & 0 & 1+2v & 0 & 0 & v & 1 & 1 & 0 & 1+2v & 2+v & 2+v \\
0 & 0 & 0 & 1+2v & v & 0 & 1+2v & 2+v & 1 & 0 & 1+v & 2+v \\
0 & 0 & 0 & 0 & 1+2v & 0 & 1+2v & 2+v & 2+2v & 1+v & v & 1+2v \\
2v & v & 2v & 0 & v & 1+v & 1 & 1+2v & 2 & 2 & 1+2v & v
\end{array}\right).
$$
is self-dual.

Here we do the same thing as in the above examples and get the code $C$
is permutation-equivalence to a code $\mathcal{C}$ generated by the following matrix:
$$
\left(\begin{array}{llllccccrrrrr}
1 & 0 & 0 & 0 & 0 & 0 & 2+v  & 2+2v & 2    & 1+2v & 0    & 2+v\\
0 & 1 & 0 & 0 & 0 & 0 & 2    & 0    & 1+2v & 2+v  & 1+2v & 2  \\
0 & 0 & 1 & 0 & 0 & 0 & 2v   & 1+2v & 1+2v & 1+2v & 2+v  & 2+2v\\
0 & 0 & 0 & 1 & 0 & 0 & 1+2v & 2+2v & v    & 2+v  & 2+v  & 2+v\\
0 & 0 & 0 & 0 & 1 & 0 & 1+2v & 1+2v & 2+v  & 2v   & 1    & 2+v\\
0 & 0 & 0 & 0 & 0 & 1 & 1+2v & 2+v  & 1+2v & 1    & 1    & 0
\end{array}\right)
=
\begin{pmatrix}
I_6 \mid vD_1+(1-v)D_2
\end{pmatrix},
$$
where
$D_1=
\begin{pmatrix}
0  & 1 & 2    & 0 & 0    & 0\\
2    & 0    & 0 & 0  & 0 & 2  \\
2   & 0 & 0 & 0 & 0  & 1\\
0 & 1 & 1    & 0  & 0  & 0\\
0 & 0 & 0  & 2   & 1    & 0\\
0 & 0  & 0 & 1    & 1    & 0
\end{pmatrix},
D_2=
\begin{pmatrix}
2  & 2 & 2    & 1 & 0    & 2\\
2    & 0    & 1 & 2  & 1 & 2  \\
0   & 1 & 1 & 1 & 2  & 2\\
1 & 2 & 0    & 2  & 2  & 2\\
1 & 1 & 2  & 0   & 1    & 2\\
1 & 2  & 1 & 1    & 1    & 0
\end{pmatrix}
$.
By the Corollary \ref{conditions-1}, it is easy to check that $\mathcal{C}$ is self-dual.
So the code $C$ is also self-dual.
\end{Example}
\noindent{\bf Acknowledgement}
This work is supported by the National Natural Science Foundation of China, Grant No. 11171370.

\medskip

\medskip

\noindent{\bf Appendix}

We give a detail proof for $HG^T=0$ in Theorem \ref{generatormatrix} below.

For $p$ being odd, we have that
\begin{align*}
HG^T & = \begin{pmatrix} v E_1+(1-v)E_2 & P & Q & I_{n-k}\\v(-A_1^T) & 0 & v I_{k_3} & 0\\
(1-v)(-B_1^T) & (1-v)I_{k_2} & 0 & 0 \end{pmatrix}
\begin{pmatrix} I_{k_1} & (1-v)B_1 & v A_1 & v D_1+(1-v)D_2  \\ 0 & v I_{k_2} & 0 & v C_1 \\
0 & 0 & (1-v)I_{k_3} & (1-v)C_2 \end{pmatrix}^T\\
& = \begin{pmatrix} v E_1+(1-v)E_2 & P & Q & I_{n-k}\\v(-A_1^T) & 0 & v I_{k_3} & 0\\
(1-v)(-B_1^T) & (1-v)I_{k_2} & 0 & 0 \end{pmatrix}
\begin{pmatrix}
I_{k_1} & 0 & 0 \\
(1-v)B^T_1 & vI_{k_2} & 0\\
vA^T_1 & 0 & (1-v)I_{k_3}\\
vD_1^T+(1-v)D^T_2 & vC_1^T & (1-v)C_2^T
\end{pmatrix}\\
& = \begin{pmatrix}
v(E_1+QA^T_1+D_1^T)+(1-v)(E_2+PB_1^T+D_2^T) & v(P+C_1^T) & (1-v)(Q+C_2^T)\\
v(-A^T_1)+v^2A^T_1 & 0 & v(1-v)I_{k_3}\\
(1-v)(-B^T_1)+(1-v)^2B^T_1 & v(1-v)I_{k_2} & 0
\end{pmatrix} \\
& = 0,
\end{align*}
where
$$v(E_1+QA^T_1+D_1^T)+(1-v)(E_2+PB_1^T+D_2^T)$$
\begin{align*}
& = v\big[ \begin{pmatrix} -A_2 \mid A_1B_4-A_3\end{pmatrix}^T
+\begin{pmatrix} 0 \mid -B_4 \end{pmatrix}^TA_1^T+D_1^T \big]
+(1-v)\big[ \begin{pmatrix} B_1A_4-B_2 \mid -B_3\end{pmatrix}^T
+\begin{pmatrix}-A_4 \mid 0 \end{pmatrix}^TB_1^T+D_2^T \big]\\
& = v\big[ \begin{pmatrix} -A_2 \mid A_1B_4-A_3\end{pmatrix}
+A_1\begin{pmatrix} 0 \mid -B_4 \end{pmatrix}+D_1 \big]^T
+(1-v)\big[ \begin{pmatrix} B_1A_4-B_2 \mid -B_3\end{pmatrix}
+B_1\begin{pmatrix}-A_4 \mid 0 \end{pmatrix}+D_2 \big]^T\\
& = v\big[ -\begin{pmatrix} A_2 \mid A_3 \end{pmatrix}+D_1\big]^T
+(1-v)\big[ -\begin{pmatrix} B_2 \mid B_3 \end{pmatrix}+D_2\big]^T\\
& = v\big( -D_1+D_1 \big)^T+(1-v)\big( -D_2+D_2\big)^T\\
& = 0;
\end{align*}
$$P+C_1^T=\begin{pmatrix}-A_4 \mid 0\end{pmatrix}^T+\begin{pmatrix}A_4 \mid 0\end{pmatrix}^T=0;$$
$$Q+C_2^T=\begin{pmatrix}0 \mid -B_4\end{pmatrix}^T+\begin{pmatrix}0 \mid B_4 \end{pmatrix}^T=0.$$

For $p=2$,
\small{
\begin{align*}
HG^T & = \begin{pmatrix}
E^TB^T+C_1^TA^T+(D_1+vD_2)^T & C_1^T & E^T & I_{n-k}\\
vB^T & 0 & vI_{k_3} & 0 \\
(1+v)A^T & (1+v)I_{k_2} & 0 & 0
\end{pmatrix}
\begin{pmatrix}
I_{k_1} & A & B & D_1+vD_2\\
0 & vI_{k_2} & 0 & vC_1\\
0 & 0 & (1+v)I_{k_3} & (1+v)E
\end{pmatrix}^T\\
& = \begin{pmatrix}
E^TB^T+C_1^TA^T+(D_1+vD_2)^T & C_1^T & E^T & I_{n-k}\\
vB^T & 0 & vI_{k_3} & 0 \\
(1+v)A^T & (1+v)I_{k_2} & 0 & 0
\end{pmatrix}
\begin{pmatrix}
I_{k_1} & 0 & 0 & \\
A^T & vI_{k_2} & 0\\
B^T & 0 & (1+v)I_{k_3}\\
D_1^T+vD_2^T & vC_1^T & (1+v)E^T
\end{pmatrix}\\
& = \begin{pmatrix}
E^TB^T+C_1^TA^T+(D_1^T+vD_2^T)+C_1^TA^T+E^TB^T+D_1^T+vD^T_2 & vC_1^T+vC_1^T & (1+v)E^T+(1+v)E^T \\
vB^T+vB^T & 0 & v(1+v)I_{k_3} \\
(1+v)A^T+(1+v)A^T & v(1+v)I_{k_2} & 0
\end{pmatrix}\\
& =0.
\end{align*}}
\normalsize
Thus we complete the proof.


\end{document}